\documentclass{article}

\usepackage{graphicx}
\usepackage{amsmath}
\usepackage{amsthm}
\usepackage{amssymb}
\usepackage{url}
\usepackage{paralist}
\usepackage{verbatim}
\usepackage{subfig}
\usepackage{wrapfig}
\usepackage{sidecap}

\usepackage{fancyhdr}
\newtheorem{theorem}{Theorem}{\bfseries}{\itshape}
\newtheorem*{theorem*}{Proof}{\bfseries}{\itshape}
\newtheorem{lemma}{Lemma}{\bfseries}{\itshape}
\newtheorem{definition}{Definition}{\bfseries}{\itshape}
\newtheorem{observation}{Observation}{\bfseries}{\itshape}

\newcounter{mycount}
\newcommand\myprob[4]{%
  \stepcounter{mycount}
  \par\noindent\\
  {} #1\\
  {\bfseries \textsc{Input}}: #2\\
  {\bfseries \textsc{Output}}: #3\\
  {}\par
}

\title{New Hardness Results for Guarding Orthogonal Polygons with Sliding
Cameras\thanks{Work supported in part by the Natural Sciences and
Engineering Research Council of Canada (NSERC).}}

\author{Stephane Durocher and Saeed Mehrabi\\
Department of Computer Science,\\ University of Manitoba, Winnipeg,
Canada.\\ \texttt{\{durocher,mehrabi\}@cs.umanitoba.ca}}

\date{}

\begin{document}

\maketitle

\begin{abstract}
Let $P$ be an orthogonal polygon. Consider a sliding camera that
travels back and forth along an orthogonal line segment $s\in P$ as
its \emph{trajectory}. The camera can see a point $p\in P$ if there
exists a point $q\in s$ such that $pq$ is a line segment normal to
$s$ that is completely inside $P$. In the \emph{minimum-cardinality
sliding cameras problem}, the objective is to find a set $S$ of
sliding cameras of minimum cardinality to guard $P$ (i.e., every
point in $P$ can be seen by some sliding camera) while in the
\emph{minimum-length sliding cameras problem} the goal is to find
such a set $S$ so as to minimize the total length of trajectories
along which the cameras in $S$ travel.

In this paper, we first settle the complexity of the minimum-length
sliding cameras problem by showing that it is polynomial tractable
even for orthogonal polygons with holes, answering a question asked
by Katz and Morgenstern~\cite{katz2011}. We next show that the
minimum-cardinality sliding cameras problem is \textsc{NP}-hard when
$P$ is allowed to have holes, which partially answers another
question asked by Katz and Morgenstern~\cite{katz2011}.
\end{abstract}

\section{Introduction}
\label{sec:introduction}%
The art gallery problem is well known in computational geometry,
where the objective is to cover a geometric shape (e.g., a polygon)
with visibility regions of a set of point guards while minimizing
the number of guards. The problem's multiple variants have been
examined extensively (e.g., see~\cite{joseph1987,urrutia2000}) and
can be classified based on the type of guards (e.g., points or line
segments), the type of visibility model and the geometric shape
(e.g., simple polygons, orthogonal polygons~\cite{hoffmann1990},
polyominoes~\cite{biedl2012}).

In this paper, we consider a variant of the orthogonal art gallery
problem introduced by Katz and Morgenstern~\cite{katz2011}, in which
\emph{sliding cameras} are used to guard the gallery. Let $P$ be an
orthogonal polygon with $n$ vertices. A sliding camera travels back
and forth along an orthogonal line segment $s$ inside $P$. The
camera (i.e., the guarding line segment $s$) can \emph{see} a point
$p\in P$ (equivalently, $p$ is \emph{orthogonally visible} to $s$)
if and only if there exists a point $q$ on $s$ such that $pq$ is
normal to $s$ and is completely contained in $P$. We study two
variants of this problem: in the \emph{minimum-cardinality sliding
cameras (MCSC) problem}, we wish to minimize the number of sliding
cameras so as to guard $P$ entirely, while in the
\emph{minimum-length sliding cameras (MLSC) problem} the objective
is to minimize the total length of trajectories along which the
cameras travel; we assume that in both variants of the problem,
polygon $P$ and sliding cameras are constrained to be orthogonal. In
both variations, every point in $P$ must be visible to some camera
at some point along its trajectory. See
Figure~\ref{fig:variantsDefinition}.

Throughout the paper, we denote an orthogonal polygon with $n$
vertices by $P$. Moreover, we denote the set of vertices and the set
of edges of $P$ by $V(P)$ and $E(P)$, respectively. We consider $P$
to be a closed set; therefore, a camera's trajectory may include an
edge of $P$. We also assume that a camera can see any point on its
trajectory. We say that a set $T$ of orthogonal line segments
contained in $P$ is a \emph{cover of} $P$, if their corresponding
cameras can collectively see any point in $P$; we sometimes say that
the line segments in $T$ \emph{guard} $P$ entirely.\\

\begin{figure}[t]
\centering%
\includegraphics[width=0.85\textwidth]{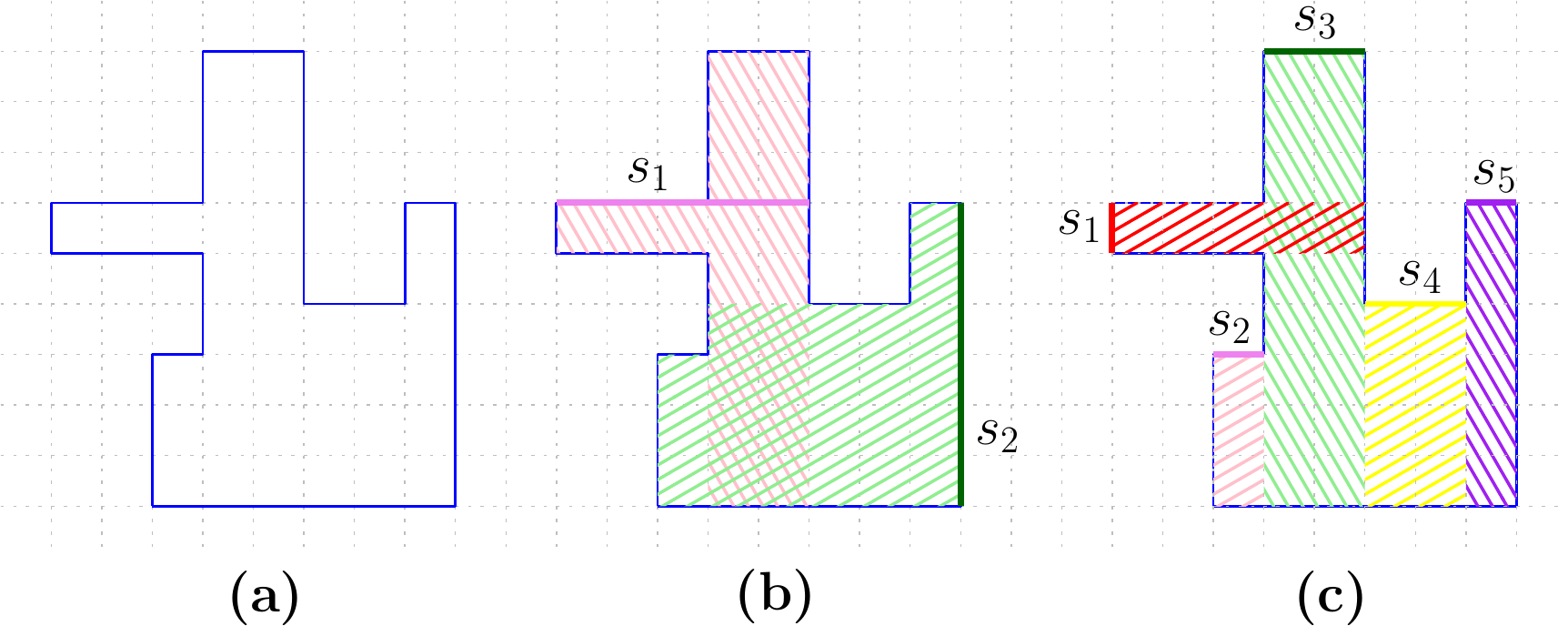}
\caption{An illustration of the variants of the problem. Each grid
cell has size $1\times1$. (a) A simple orthogonal polygon $P$. (b)
An optimal solution for the minimum-cardinality sliding cameras
problem on $P$. The trajectories of two sliding cameras $s_1$ and
$s_2$ are shown in pink and green, respectively; each shaded region
indicates the visibility region of the corresponding camera. This
set of two cameras is an optimal solution to the minimum-cardinality
sliding cameras problem on $P$. (c) A set of five sliding cameras
whose total length of trajectories is 7, which is an optimal
solution for the minimum-length sliding cameras problem on $P$.}
\label{fig:variantsDefinition}%
\end{figure}

\noindent{\bf Related Work.} The art gallery problem was first
introduced by Klee in 1973. Two years later,
Chvatal~\cite{chvatal1975} gave an upper bound proving that $\lfloor
n/3 \rfloor$ point guards are always sufficient and sometimes
necessary to guard a simple polygon with $n$ vertices. The
orthogonal art gallery problem was first studied by Kahn et
al.~\cite{jeff1983} who proved that $\lfloor n/4 \rfloor$ guards are
always sufficient and sometimes necessary to guard the interior of a
simple orthogonal polygon. Lee and Lin~\cite{lee1986} showed that
the problem of guarding a simple polygon using the minimum number of
guards is \textsc{NP}-hard. Moreover, the problem was also shown to
be \textsc{NP}-hard for orthogonal polygons~\cite{dietmar1995}.

Limiting visibility allows some versions of the problem to be solved
in polynomial time. Motwani et al.~\cite{motwani1988} studied the
art gallery problem under $s$-visibility, where a guard point $p\in
P$ can see all points in $P$ that can be connected to $p$ by an
orthogonal staircase path contained in $P$. They use a perfect graph
approach to solve the problem in polynomial time. Worman and
Keil~\cite{worman2007} defined $r$-visibility, in which a guard
point $p\in P$ can see all points $q\in P$ such that the bounding
rectangle of $p$ and $q$ (i.e., the axis-parallel rectangle with
diagonal $\overline{pq}$) is contained in $P$. Given that $P$ has
$n$ vertices, they use a similar approach to Motwani et
al.~\cite{motwani1988} to solve this problem in
$\widetilde{O}(n^{17})$ time, where $\widetilde{O}()$ hides
poly-logarithmic factors. Moreover, Lingas et al.~\cite{lingas2008}
presented a linear-time 3-approximation algorithm for this problem.

Recently, Katz and Morgenstern~\cite{katz2011} introduce sliding
cameras as another model of visibility to guard a \emph{simple}
orthogonal polygon $P$; they only study the MCSC problem. They first
consider a restricted version of the problem, where cameras are
constrained to travel only vertically inside the polygon. Using a
similar approach to Motwani et al.~\cite{motwani1988} they construct
a graph $G$ corresponding to $P$ and then show that
\begin{inparaenum}[(i)]\item solving this problem on $P$ is
equivalent to solving the minimum clique cover problem on $G$, and
that \item $G$ is chordal.
\end{inparaenum} Since the minimum clique cover problem is
polynomial solvable on chordal graphs, they solve the problem in
polynomial time. They also generalize the problem such that both
vertical and horizontal cameras are allowed (i.e., the MCSC
problem); they present a 2-approximation algorithm for this problem
under the assumption that the given input is an $x$-monotone
orthogonal polygon. They leave open the complexity of the problem
and mention studying the minimum-length sliding cameras problem as
future work.

A \emph{histogram} $H$ is a simple polygon that has an edge, called
the \emph{base}, whose length is equal to the sum of the lengths of
the edges of $H$ that are parallel to the base. Moreover, a
\emph{double-sided histogram} is the union of two histograms that
share the same base edge and that are located on opposite sides of
the base. It is easy to observe that the MCSC problem is equivalent
to the problem of covering $P$ with minimum number of double-sided
histograms. Fekete and Mitchell~\cite{fekete2001} proved that
partitioning an orthogonal polygon (possibly with holes) into a
minimum number of histograms is \textsc{NP}-hard. However, their
proof does not directly imply that the MCSC problem is also
\textsc{NP}-hard for orthogonal polygons with holes.\\

\noindent{\bf Our Results.} In this paper, we first answer a
question asked by Katz and Morgenstern~\cite{katz2011} by proving
that the MLSC problem is solvable in polynomial time even for
orthogonal polygons with holes (see
Section~\ref{sec:optimalAlgorithmLength}). We next show that the
MCSC problem is \textsc{NP}-hard for orthogonal polygons with holes
(see Section~\ref{sec:minimumCardinality}), which partially answers
another question asked by Katz and Morgenstern~\cite{katz2011}. We
conclude the paper by Section~\ref{sec:conclusion}.

\section{The MLSC Problem: An Exact Algorithm}
\label{sec:optimalAlgorithmLength}%
In this section, we give an algorithm that solves the MLSC problem
exactly in polynomial time even when $P$ has holes. Let $T$ be a
cover of $P$. In this section, we say that $T$ is an \emph{optimal
cover for $P$} if the total length of trajectories along which the
cameras in $T$ travel is minimum over that of all covers of $P$. Our
algorithm relies on reducing the MLSC problem to the
\emph{minimum-weight vertex cover problem} in graphs. We remind the
reader of the definition of the minimum-weight vertex cover problem:
\begin{definition}
\label{def:weightVertexCover}%
Given a graph $G=(V,E)$ with positive edge weights, the {\em
minimum-weight vertex cover} problem is to find a subset $V'
\subseteq V$ that is a vertex cover of $G$ (i.e., every edge in $E$
has at least one endpoint in $V'$) such that the sum of the weights
of vertices in $V'$ is minimized.
\end{definition}
The minimum-weight vertex cover problem is \textsc{NP}-hard in
general~\cite{karp1972}. However, it is solvable in polynomial time
when the input graph is \emph{bipartite} because the constraint
matrix of the Integer Program corresponding to the minimum-weight
vertex cover problem is totally unimodular~\cite{wikipediaVC}. Given
$P$, we first construct a vertex-weighted graph $G_P$ and then we
show that \begin{inparaenum}[(i)]\item the MLSC problem on $P$ is
equivalent to the minimum-weight vertex cover problem on $G_P$, and
that \item graph $G_P$ is bipartite. \end{inparaenum}\\

Similar to Katz and Morgenstern~\cite{katz2011}, we define a
partition of an orthogonal polygon $P$ into rectangles as follows.
Extend the two edges of $P$ incident to every reflex vertex in
$V(P)$ inward until they hit the boundary of $P$. Let $S(P)$ be the
set of the extended edges and the edges of $P$ whose endpoints are
both non-reflex vertices of $P$. We refer to elements of $S(P)$
simply as \emph{edges}. The edges in $S(P)$ partition $P$ into a set
of rectangles; let $R(P)$ denote the set of resulting rectangles. We
observe that in order to guard $P$ entirely, it suffices to guard
all rectangles in $R(P)$. The following observations are
straightforward:

\begin{observation}
\label{obs:smallerParts}%
Let $T$ be a cover of $P$ and let $s$ be an orthogonal line segment
in $T$. Then, for any partition of $s$ into line segments
$s_1,s_2,\ldots,s_k$ the set $T'=(T\setminus s)\cup\{s_1,\ldots,
s_k\}$ is also a cover of $P$ and the respective sums of the lengths
of segments in $T$ and $T'$ are equal.
\end{observation}

\begin{observation}
\label{obs:movement}%
Let $T$ be a cover of $P$. Moreover, let $T'$ be the set of line
segments obtained from $T$ by translating every vertical line
segment in $T$ horizontally to the nearest boundary of $P$ to its
right and every horizontal line segment in $T$ vertically to the
nearest boundary of $P$ below it. Then, $T'$ is also a cover of $P$
and the respective sums of the lengths of line segments in $T$ and
$T'$ are equal. We call $T'$ a \emph{regular cover} of $P$.
\end{observation}

We now prove the following result.

\begin{lemma}
\label{lem:guardingARectangle}%
Let $R\in R(P)$ be a rectangle and let $T$ be a cover of $P$. Then,
there exists a set $T'\subseteq T$ such that all line segments in
$T'$ have the same orientation (i.e., they are all vertical or they
are all horizontal) and they collectively guard $R$ entirely.
\end{lemma}
\begin{proof}
Suppose, by a contradiction, that there does not exists such a set
$T'$. Let $R_v$ (resp., $R_h$) be the subregion of $R$ that is
guarded by the all union of the vertical (resp., horizontal) line
segments in $T$ and let $R_v^c=R\setminus R_v$ (resp.,
$R_h^c=R\setminus R_h$). Since $R$ cannot be guarded exclusively by
vertical line segments (resp., horizontal line segments), we have
$R_v^c\ne\emptyset$ (resp., $R_h^c\ne\emptyset$). Choose any point
$p\in R_v^c$ and let $L_h$ be the maximal horizontal line segment
inside $R$ that crosses $p$. Since no vertical line segment in $T$
can guard $p$, we conclude that no point on $L_h$ is guarded by a
vertical line segment in $T$. Similarly, choose any point $q\in
R_h^c$ and let $L_v$ be the maximal vertical line segment inside $R$
that contains $q$. By an analogous argument, we conclude that no
point on $L_v$ is guarded by a horizontal line segment. Since $L_h$
and $L_v$ are maximal and have perpendicular orientations, $L_h$ and
$L_v$ intersect inside $R$. Therefore, no orthogonal line segment in
$T$ can guard the intersection point of $L_h$ and $L_v$, which is a
contradiction.
\end{proof}

Given $P$, let $H(P)$ denote the subset of the boundary of $P$
consisting of line segments that are immediately to the right of or
below $P$. Let $B(P)$ denote the partition of $H(P)$ into line
segments induced by the edges in $S(P)$. The following lemma follows
by Lemma~\ref{lem:guardingARectangle} and
Observations~\ref{obs:smallerParts} and~\ref{obs:movement}:
\begin{lemma}
\label{lem:minLengthExistence}%
Every orthogonal polygon $P$ has an optimal cover $T\subseteq B(P)$.
\end{lemma}

\begin{observation}
\label{obs:oneToOneCorrespondence}%
Let $P$ be an orthogonal polygon and consider its corresponding set
$R(P)$ of rectangles induced by edges in $S(P)$. Every rectangle
$R\in R(P)$ is seen by exactly one vertical line segment in $B(P)$
and exactly one horizontal line segment in $B(P)$. Furthermore, if
$T\subseteq B(P)$ is a cover of $P$, then every rectangle in $R(P)$
must be seen by at least one horizontal or one vertical line segment
in $T$.
\end{observation}

We denote the horizontal and vertical line segments in $B(P)$ that
can see a rectangle $R\in R(P)$ by $R_V$ and $R_H$, respectively.
Observation~\ref{obs:oneToOneCorrespondence} leads us to reducing
the problem to the minimum-weight vertex cover problem on graphs. We
construct an undirected weighted graph $G_P=(V, E)$ associated with
$P$ as follows: each line segment $s\in B(P)$ corresponds to a
vertex $v_s\in V$ such that the weight of $v_s$ is the length of $s$
(we denote the vertex in $V$ that corresponds to the line segment
$s\in B(P)$ by $v_s$). Two vertices $v_s, v_{s'}\in V$ are adjacent
in $G_P$ if and only if the line segments $s$ and $s'$ can both see
a fixed rectangle $R\in R(P)$. See Figure~\ref{fig:reductionExample}
for an illustration of the reduction. By
Observation~\ref{obs:oneToOneCorrespondence} the following result is
straightforward:
\begin{observation}
\label{obs:rectanglesEdgesCorrespondence}%
For each rectangle $R\in R(P)$, there exists exactly one edge in
$G_P$ that is correspond to $R$.
\end{observation}

\begin{figure}[t]
\centering%
\includegraphics[width=0.80\textwidth]{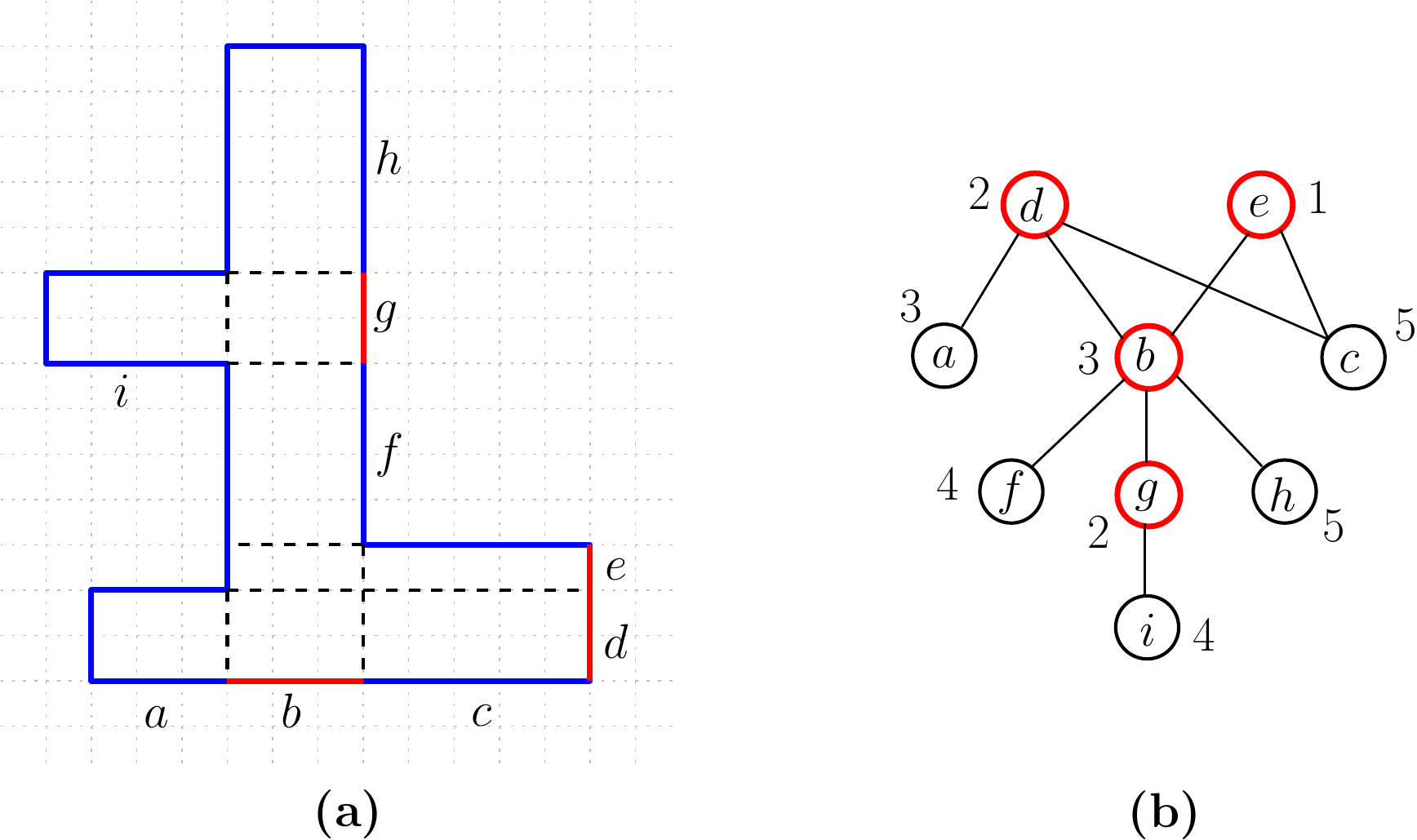}
\caption{An illustration of the reduction; each grid cell has size
$1\times1$. (a) An orthogonal polygon $P$ along with the elements of
$E(P)\cup S(P)$ labeled as $a,b,c,\dots,i$. (b) The graph $G_P$
associated with $P$; the integer value besides each vertex indicates
the weight of the vertex. The vertices of a minimum-weight vertex
cover on $G_P$ and their corresponding guarding line segments for
$P$ are shown in red.}
\label{fig:reductionExample}%
\end{figure}

We now prove the following result:
\begin{theorem}
\label{thm:minLengthVCoverEquiv}%
The MLSC problem on $P$ reduces to the minimum-weight vertex cover
problem on $G_P$.
\end{theorem}
\begin{proof}
Let $S_0$ be a vertex cover of $G_P$ and let $C_0$ be a cover of $P$
defined in terms of $S_0$; the mapping from $S_0$ to $C_0$ will be
defined later. Moreover, for each vertex $v$ of $G_P$ let $w(v)$
denote the weight of $v$ and for each line segment $s\in C_0$ let
$len(s)$ denote the length of $s$. We need to prove that $S_0$ is a
minimum-weight vertex cover of $G_P$ if and only if $C_0$ is an
optimal cover of $P$. We show the following stronger statements:
\begin{itemize}
\item for any vertex cover $S$ of $G_P$, there exists a cover $C$ of $P$
such that $\sum_{s\in C}len(s)=\sum_{v\in S}w(v)$, and
\item for any cover $C$ of $P$, there exists a vertex cover $S$
of $G_P$ such that $\sum_{v\in S}w(v)=\sum_{s\in C}len(s)$.
\end{itemize}
\noindent{\bf Part 1.} Choose any vertex cover $S$ of $G_P$. We find
a cover $C$ for $P$ as follows: for each edge $(v_s, v_{s'})\in E$,
if $v_s\in S$ we locate a guarding line segment on the boundary of
$P$ that is aligned with the line segment $s\in B(P)$. Otherwise, we
locate a guarding line segment on the boundary of $P$ that is
aligned with the line segment $s'\in B(P)$. Since at least one of
$v_s$ and $v_{s'}$ is in $S$, we conclude by
Observation~\ref{obs:rectanglesEdgesCorrespondence} that every
rectangle in $R(P)$ is guarded by at least one line segment located
on the boundary of $P$ and so $C$ is a cover of $P$. Moreover, for
each vertex in $S$ we locate exactly one guarding line segment on
the boundary of $P$ whose length is the same as the weight of the
vertex. Therefore, $\sum_{s\in C}len(s)=\sum_{v\in S}w(v)$.

\noindent{\bf Part 2.} Choose any cover $C$ of $P$. We construct a
vertex cover $S$ for $G_P$ as follows. By
Observation~\ref{obs:movement}, let $T'$ be the regular cover
obtained from $C$. Moreover, let $M$ be the partition of $T'$ into
line segments induced by the edges in $S(P)$. By
Lemma~\ref{lem:guardingARectangle}, for any rectangle $R\in R(P)$,
there exists a set $C'_R\subseteq C$ such that all line segments in
$C'_R$ have the same orientation and collectively guard $R$.
Therefore, $M$ is also a cover of $P$. Now, let $S$ be the subset of
the vertices of $G_P$ such that $v_s\in S$ if and only if $s\in M$.
Since $M$ is a cover of $G_P$ we conclude, by
Observation~\ref{obs:rectanglesEdgesCorrespondence}, that $S$ is a
vertex cover of $G_P$. Moreover, we observe that the total weight of
the vertices in $S$ is the same as the total length of the line
segments in $M$ and, therefore, $\sum_{v\in S}w(v)=\sum_{s\in
C}len(s)$.
\end{proof}
We next show that the graph $G_P$ is bipartite.
\begin{lemma}
\label{lem:gIsBipartite}%
Graph $G_P$ is bipartite.
\end{lemma}
\begin{proof}
The proof follows from the facts that \begin{inparaenum}[(i)]\item
we have two types of edges in $G_P$; those that correspond to the
vertical line segments in $B(P)$ and those that correspond to the
horizontal line segments in $B(P)$, and that \item no two vertical
line segments in $B(P)$ nor any two horizontal line segments in
$B(P)$ can see a fixed rectangle in $R(P)$.\end{inparaenum}
\end{proof}
It is easy to see that the construction in the proof of
Theorem~\ref{thm:minLengthVCoverEquiv} can be completed in
polynomial time. Therefore, by
Theorem~\ref{thm:minLengthVCoverEquiv}, Lemma~\ref{lem:gIsBipartite}
and the fact that minimum-weight vertex cover is solvable in
polynomial time on bipartite graphs \cite{wikipediaVC}, we have the
main result of this section:
\begin{theorem}
\label{thm:optimalAlgorithmForMinLength}%
Given an orthogonal polygon $P$ with $n$ vertices, there exists an
algorithm that finds an optimal cover of $P$ in time polynomial in
$n$.
\end{theorem}

\section{The MCSC Problem: \textsc{NP}-hardness Result}
\label{sec:minimumCardinality}%
In this section, we show that the following problem is
\textsc{NP}-hard.

\myprob{\textsc{MCSC With Holes}}
{An orthogonal polygon $P$ possibly with holes.}%
{An optimal solution for the MCSC problem on $P$.}

We show \textsc{NP}-hardness by a reduction from the \emph{minimum
hitting of horizontal unit segments} problem, which we call it the
\textsc{Min Segment Hitting} problem. The \textsc{Min Segment
Hitting} problem is defined as follows.
\begin{definition}[Hassin and Meggido~\cite{hassin1991},
1991]
\label{def:minimumHitProblem}%
Given $n$ pairs $(a_i, b_i)$, $i=1, \dots, n$, of integers and an
integer $k$, decide whether there exist $k$ orthogonal lines $l_1,
\dots, l_k$ in the plane such that each line segment $[(a_i, b_i),
(a_i+1, b_i)]$ is hit by at least one of the lines.
\end{definition}
Hassin and Meggido~\cite{hassin1991} prove that the \textsc{Min
Segment Hitting} problem is \textsc{NP}-complete. Let $I$ be an
instance of the \textsc{Min Segment Hitting} problem, where $I$ is a
set of $n$ horizontal unit-length segments. We construct an
orthogonal polygon $P$ (with holes) such that there exists a set of
$k$ orthogonal lines that hit the segments in $I$ if and only if
there exists a set $C$ of $k+1$ orthogonal line segments inside $P$
that collectively guard $P$. Throughout this section, we refer to
the segments in $I$ as \emph{unit segments} and to the segments in
$C$ as \emph{line segments}.\\

\begin{figure}[t]
\centering%
\includegraphics[width=0.70\textwidth]{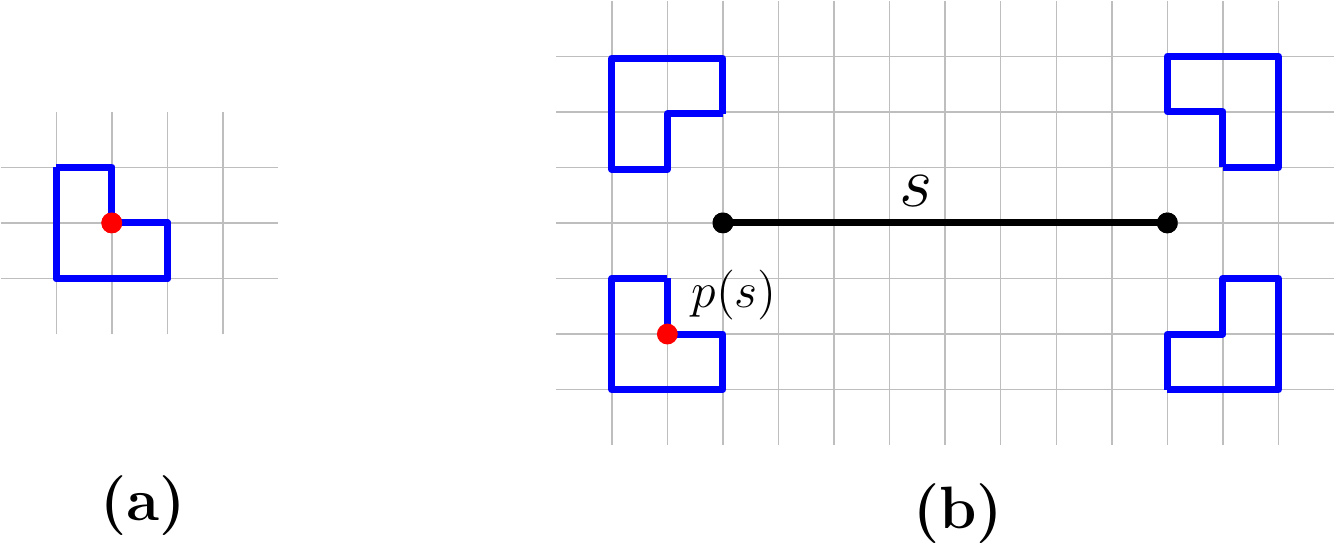}
\caption{(a) An $L$-hole gadget; each grid cell has size
$\frac{1}{12}\times\frac{1}{12}$. (b) The $L$-holes associated with
a line segment $s\in I$, where the $x$-coordinate of $a_s$ is even.}
\label{fig:initialGadgets}%
\end{figure}

\noindent{\bf Gadgets.} We first observe that any two unit segments
in $I$ can share at most one point, which must be a common endpoint
of the two unit segments. Let $s$ be a unit segment in $I$. We
denote the left and right endpoints of $s$ by $a_s$ and $b_s$,
respectively. Moreover, let $N(s)$ denote the set of unit segments
in $I$ that have at least one endpoint with $x$-coordinate equal to
that of $a_s$ or $b_s$. Our reduction refers to an $L$-hole, which
we define as a minimum-area orthogonal polygon with vertices at grid
coordinates such that exactly one of which is a reflex vertex.
Figure~\ref{fig:initialGadgets}(a) shows an $L$-hole. We constrain
each grid cell to have size $\frac{1}{12}\times\frac{1}{12}$. An
$L$-hole may be rotated by $\pi/2$, $\pi$ or $3\pi/2$. For each unit
segment $s\in I$, we associate exactly four $L$-holes with $s$
depending on the parity of the $x$-coordinate of $a_s$:
\begin{itemize}
\item If the $x$-coordinate of $a_s$ is even, then
Figure~\ref{fig:initialGadgets}(b) shows the $L$-holes associated
with $s$. The $L$-holes associated with $s$ do not interfere with
the $L$-holes associated with the line segments in $N(s)$ because
the unit segments in $N(s)$ have the vertical distance at least one
to $s$. Note the red vertex on the bottom left $L$-hole of $s$; we
call this vertex the \emph{visibility vertex} of $s$, which we
denote $p(s)$.
\item If the $x$-coordinate of $a_s$ is odd, then
Figure~\ref{fig:holesForOddUnitSegments} shows the $L$-holes
associated with $s$. Note that, in this case, the $L$-holes are
located such that the vertical distance between any point on an
$L$-hole and $s$ is at least $3/12$. By an analogous argument, we
observe that the $L$-holes associated with $s$ do not interfere with
the $L$-holes associated with the line segments in $N(s)$. Note the
blue vertex on the bottom right $L$-hole of $s$; we call this vertex
the \emph{visibility vertex} of $s$, which we denote $p(s)$.
\end{itemize}

Let $s$ and $s'$ be two unit segments in $I$ that share a common
endpoint. Since $s$ and $s'$ have unit lengths the $x$-coordinates
of $a_s$ and $a_{s'}$ have different parities. Therefore, the
$L$-holes associated with $s$ and $s'$ do not interfere with each
other. Figure~\ref{fig:multipleSegments} shows an example of two
unit segments $s$ and $s'$ and their corresponding $L$-holes. We now
describe the reduction.\\

\noindent{\bf Reduction.} Given an instance $I$ of the \textsc{Min
Segment Hitting} problem, we first associate each unit segment in
$s\in I$ with four $L$-holes depending on whether the $x$-coordinate
of $a_s$ is even or odd. After adding the corresponding $L$-holes,
we enclose $I$ in a rectangle such that the all unit segments and
the $L$-holes associated with them lie in its interior. Finally, we
create a small rectangle on the bottom left corner of the bigger
rectangle (see Figure~\ref{fig:reductionExample}) such that any
orthogonal line that passes through the smaller rectangle cannot
intersect any of the unit segments in $I$. See
Figure~\ref{fig:reductionExample} for a complete example of the
reduction. Let $P$ be the resulting orthogonal polygon. We first
have the following observation.

\begin{figure}[t]
\centering%
\subfloat[The $L$-holes associated with a line segment $s\in I$,
where the $x$-coordinate of $a_s$ is odd.]
{\label{fig:holesForOddUnitSegments}
\includegraphics[width=0.37\textwidth]{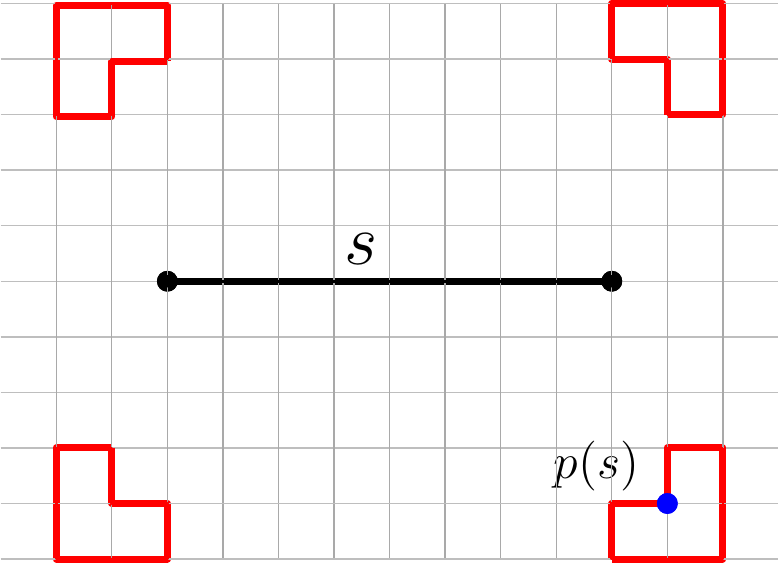}}
~\quad%
\subfloat[An illustration of the $L$-holes associated with two line
segments in $I$ that share a common endpoint.]
{\label{fig:multipleSegments}
\includegraphics[width=0.50\textwidth]{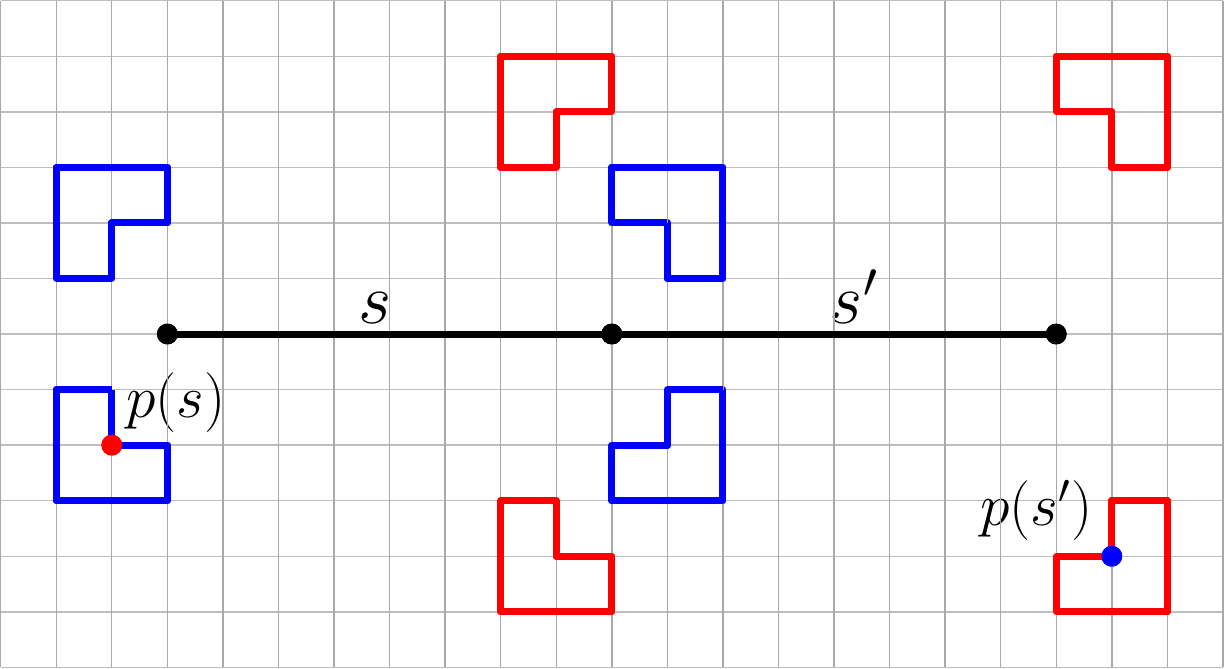}}
\caption{An illustration of the gadgets used in the reduction.}
\label{fig:oneSegment}%
\end{figure}

\begin{observation}
\label{obs:verticalOnlyOneGuarding}%
Let $s$ be a unit segment in $I$. Moreover, let $l$ be a vertical
line segment contained in $P$ that can see $p(s)$. If $l$ does not
intersect $s$, then $p(s')$ is not orthogonally visible to $l$ for
all $s'\in I\setminus \{s\}$.
\end{observation}
We now show the following result.
\begin{lemma}
\label{lem:npHardnessIff}%
There exist $k$ orthogonal lines such that each unit segment in $I$
is hit by one of the lines if and only if there exists $k+1$
orthogonal line segments contained in $P$ that collectively guard
$P$.
\end{lemma}

\begin{figure}[t]
\centering%
\includegraphics[width=0.60\textwidth]{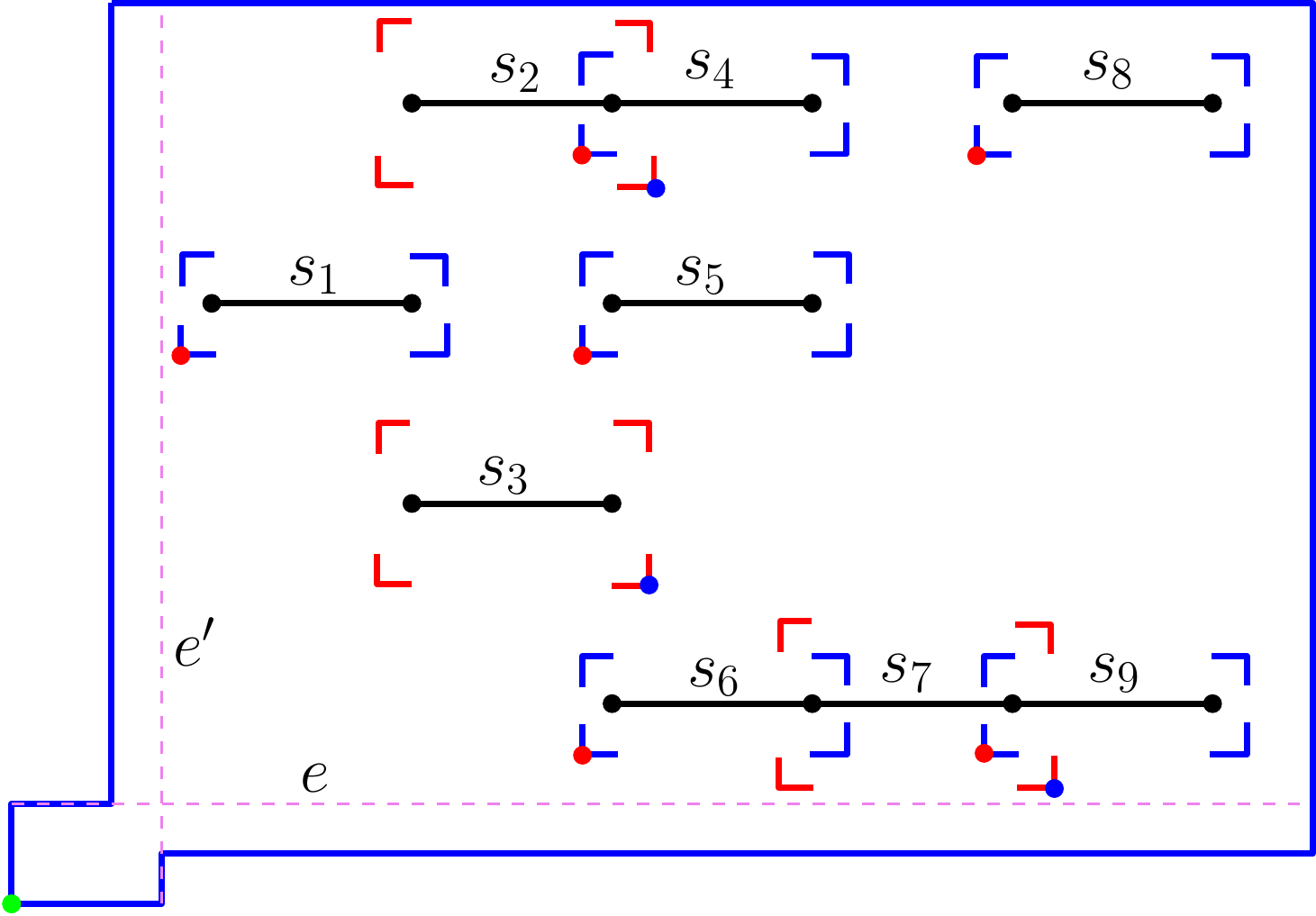}
\caption{A complete example of the reduction, where $I=\{s_1,
s_2,\dots,s_9\}$, with the assumption that the $x$-coordinate of
$a_{s_1}$ is even. Each line segment that has a bend represents an
$L$-hole associated with a unit segment. The visibility vertices of
the unit segments in $I$ are shown red or blue appropriately. Note
the green vertex on the lower left corner of the smaller rectangle;
this vertex is only visible to the line segments that pass through
the interior of the smaller rectangle, which in turn cannot
intersect any unit segment in $I$.}
\label{fig:reductionExample}%
\end{figure}

\begin{proof}
$(\Rightarrow)$ Suppose there exists a set $S$ of $k$ lines such
that each unit segment in $I$ is hit by at least one line in $S$.
Let $L\in S$ and let $L_P=L\cap P$. If $L$ is horizontal, then it is
easy to see that $L$, and therefore $L_P$, does not cross any
$L$-hole inside $P$. Similarly, if $L$ is vertical and passes
through an endpoint of some unit segment(s) in $I$, then neither $L$
nor $L_P$ passes through the interior of any $L$-hole in
$P$.\footnote{Note that it is possible that $L$ passes through the
boundary of some $L$-hole.} Now, suppose that $L$ is vertical and
passes through the interior of some unit segment $s\in I$. Translate
$L_P$ horizontally such that it passes through the midpoint of $s$.
Since unit segments have endpoints on adjacent integer grid point,
$L_P$ still crosses the same set of unit segments of $I$ as it did
before this move. Moreover, this ensures that $L_P$ does not cross
any $L$-hole inside $P$. Consider the set $S'=\{L_P\mid L\in S\}$.

We observe that the line segments in $S'$ cannot guard the interior
of the smaller rectangle. Moreover, if the all line segments in $S'$
are vertical or all of them are horizontal, then they cannot
collectively guard the bigger rectangle
entirely.\footnote{Specifically, in either cases, there are regions
between two $L$-holes associated with different unit segments that
cannot be guarded by any line segment.} In order to guard $P$
entirely, we add one more orthogonal line segment $C$ as follows: if
the all line segments in $S'$ are vertical (resp., horizontal), then
$C$ is the maximal horizontal (resp., the maximal vertical) line
segment inside $P$ that aligns with the upper edge (resp., the right
edge) of the smaller rectangle of $P$; see the line segment $e$
(resp., $e'$) in Figure~\ref{fig:reductionExample}. If the line
segments in $S'$ are a combination of vertical and horizontal line
segments, then $C$ can be either $e$ or $e'$. It is easy to observe
that now the line segments in $S'$ along with $C$ collectively guard
$P$ entirely. Therefore, we have established that the entire polygon
$P$ is guarded by $k+1$ orthogonal line segments inside $P$ in
total.

$(\Leftarrow)$ Now, suppose that there exists a set $M$ of $k+1$
orthogonal line segments contained in $P$ that collectively guard
$P$. Let $c\in M$ and let $L_c$ denote the line induced by $c$. We
find $k$ lines that form a solution to instance $I$ by moving the
line segments in $M$ accordingly such that each unit segment in $I$
is hit by at least one of the corresponding lines. Let $c_0\in M$ be
the line segment that guards the bottom left vertex of the smaller
rectangle of $P$. We know that $L_{c_0}$ cannot guard $p(s)$ for all
line segments $s\in I$. For each unit segment $s\in I$ in order,
consider a line segment $l\in M\setminus \{c_0\}$ that guards
$p(s)$. If $l$ is horizontal and $L_l$ does not align $s$, then move
$l$ accordingly up or down until it aligns $s$. Therefore, $L_l$ is
a line that hits $s$. Now, suppose that $l$ is vertical. If $l$
intersects $s$\footnote{It might be possible that a line segment
guards the visibility vertex of a unit segment but does not
intersect the unit segment.}, then $L_l$ also intersects $s$.
Otherwise, by Observation~\ref{obs:verticalOnlyOneGuarding}, $p(s)$
is the only visibility vertex that is visible to $l$. Therefore,
move $l$ horizontally to the left or to the right until it hits $s$.
Therefore, $L_l$ is a line that hits $s$ after this move.

We observe that we obtained exactly one line from each line segment
in $M\setminus \{c_0\}$. Therefore, we have found $k$ lines such
that each unit segment in $I$ is hit by at least one of the lines.
This completes the proof of the lemma.
\end{proof}
By Lemma~\ref{lem:npHardnessIff}, we obtain the main result of this
section:
\begin{theorem}
\label{subsec:orthogonalCameras:npHardnessPolygonsWithHoles}%
The \textsc{MCSC With Holes} is \textsc{NP}-hard.
\end{theorem}

\section{Conclusion}
\label{sec:conclusion}%
In this paper, we studied the problem of guarding an orthogonal
polygon $P$ with sliding cameras that was introduced by Katz and
Morgenstern~\cite{katz2011}. We considered two variants of this
problem: the minimum-cardinality sliding cameras problem (in which
the objective is to minimize the number of sliding cameras used to
guard $P$) and the minimum-length sliding cameras problem (in which
the objective is to minimize the total length of trajectories along
which the cameras travel).

We gave a polynomial-time algorithm that solves the minimum-length
sliding cameras problem exactly even for orthogonal polygons with
holes, answering a question asked by Katz and
Morgenstern~\cite{katz2011}. We also showed that the
minimum-cardinality sliding cameras problem is \textsc{NP}-hard when
$P$ contains holes, which partially answers another question asked
by Katz and Morgenstern~\cite{katz2011}.



\begin{thebibliography}{10}

\bibitem{biedl2012}
Therese~C. Biedl, Mohammad~Tanvir Irfan, Justin Iwerks, Joondong Kim, and
  Joseph S.~B. Mitchell.
\newblock The art gallery theorem for polyominoes.
\newblock {\em Discrete {\&} Computational Geometry}, 48(3):711--720, 2012.

\bibitem{chvatal1975}
Vasek Chvatal.
\newblock A combinatorial theorem in plane geometry.
\newblock {\em Journal of Combinatorial Theory, Series B}, 18:39--41, 1975.

\bibitem{fekete2001}
S\'{a}ndor~P. Fekete and Joseph S.~B. Mitchell.
\newblock Terrain decomposition and layered manufacturing.
\newblock {\em International Journal of Computational Geometry and
  Applications}, 11(6):647--668, 2001.

\bibitem{hassin1991}
Refael Hassin and Nimrod Megiddo.
\newblock Approximation algorithms for hitting objects with straight lines.
\newblock {\em Discrete Applied Mathematics}, 30(1):29--42, 1991.

\bibitem{hoffmann1990}
Frank Hoffmann.
\newblock On the rectilinear art gallery problem.
\newblock In {\em 17th International Colloquium on Automata, Languages and
  Programming, (ICALP)}, pages 717--728, 1990.

\bibitem{jeff1983}
Jeff Kahn, Maria~M. Klawe, and Daniel~J. Kleitman.
\newblock Traditional galleries require fewer watchmen.
\newblock {\em SIAM Journal on Algebraic Discrete Methods}, 4(2):194--206,
  1983.

\bibitem{karp1972}
Richard~M. Karp.
\newblock Reducibility among combinatorial problems.
\newblock In {\em Complexity of Computer Computations}, pages 85--103, 1972.

\bibitem{katz2011}
Matthew~J. Katz and Gila Morgenstern.
\newblock Guarding orthogonal art galleries with sliding cameras.
\newblock {\em International Journal of Computational Geometry and
  Applications}, 21(2):241--250, 2011.

\bibitem{lee1986}
D.~T. Lee and Arthur~K. Lin.
\newblock Computational complexity of art gallery problems.
\newblock {\em IEEE Transactions on Information Theory}, 32(2):276--282, 1986.

\bibitem{lingas2008}
Andrzej Lingas, Agnieszka Wasylewicz, and Pawel Zylinski.
\newblock Linear-time 3-approximation algorithm for the {\it r}-star covering
  problem.
\newblock In {\em Worshop on Algorithms and Computation}, pages 157--168, 2008.

\bibitem{motwani1988}
Rajeev Motwani, Arvind Raghunathan, and Huzur Saran.
\newblock Covering orthogonal polygons with star polygons: the perfect graph
  approach.
\newblock In {\em Symposium on Computational Geometry (SoCG)}, pages 211--223,
  1988.

\bibitem{joseph1987}
Joseph O'Rourke.
\newblock {\em Art gallery theorems and algorithms}.
\newblock Oxford University Press, 1987.

\bibitem{dietmar1995}
Dietmar Schuchardt and Hans-Dietrich Hecker.
\newblock Two \textsc{NP}-hard art-gallery problems for ortho-polygons.
\newblock {\em Mathematical Logic Quarterly}, 41(2):261--267, 1995.

\bibitem{urrutia2000}
Jorge Urrutia.
\newblock Art gallery and illumination problems.
\newblock In {\em Handbook of Computational Geometry}, pages 973--1027.
  North-Holland, 2000.

\bibitem{wikipediaVC}
Wikipedia.
\newblock Vertex cover.
\newblock \url{http://en.wikipedia.org/wiki/Vertex_cover}.

\bibitem{worman2007}
Chris Worman and J.~Mark Keil.
\newblock Polygon decomposition and the orthogonal art gallery problem.
\newblock {\em International Journal of Computational Geometry and
  Applications}, 17(2):105--138, 2007.

\end{thebibliography}
\end{document}